\documentclass[letterpaper]{article} %
\usepackage{aaai24}  %
\usepackage{times}  %
\usepackage{helvet}  %
\usepackage{courier}  %
\usepackage[hyphens]{url}  %
\usepackage{graphicx} %
\usepackage{natbib}  %
\usepackage{caption} %
\usepackage{algorithm}
\usepackage{algorithmic}

\usepackage{newfloat}
\usepackage{listings}
\DeclareCaptionStyle{ruled}{labelfont=normalfont,labelsep=colon,strut=off} %
\floatstyle{ruled}
\newfloat{listing}{tb}{lst}{}
\floatname{listing}{Listing}
\title{Distributional Off-Policy Evaluation for Slate Recommendations}
\author{
    Shreyas Chaudhari\textsuperscript{\rm 1}\thanks{Correspondence to: schaudhari@cs.umass.edu},
    David Arbour\textsuperscript{\rm 2},
    Georgios Theocharous\textsuperscript{\rm 2},
    Nikos Vlassis\textsuperscript{\rm 2}
}
\affiliations{
    \textsuperscript{\rm 1}University of Massachusetts Amherst\\
    \textsuperscript{\rm 2}Adobe Research
}

\usepackage{bibentry}

\usepackage{amsmath}
\usepackage{amsthm}
\usepackage{booktabs}
\usepackage{algorithm}
\usepackage{algorithmic}
\usepackage{xcolor}
\usepackage{amsfonts}
\usepackage{thmtools,thm-restate}

\usepackage{enumitem}
\newtheorem{theorem}{Theorem}

\newtheorem{ass}[theorem]{Assumption}

\newcommand{\E}{\mathbb{E}}
\newcommand{\I}{\mathbf{1}}
\newcommand{\Var}{\mathrm{Var}}

\begin{document}

\maketitle

\begin{abstract}

Recommendation strategies are typically evaluated by using previously logged data, employing off-policy evaluation methods to estimate their expected performance.
However, for strategies that present users with slates of multiple items, the resulting combinatorial action space renders many of these methods impractical. %
Prior work has developed estimators that leverage the structure in slates to estimate the expected off-policy performance, but the estimation of the entire performance \textit{distribution} remains elusive.
Estimating the complete distribution allows for a more comprehensive evaluation of recommendation strategies, particularly along the axes of risk and fairness that employ metrics computable from the distribution. In this paper, we propose an estimator for the complete off-policy performance distribution for slates and establish conditions under which the estimator is unbiased and consistent.
This builds upon prior work on off-policy evaluation for slates and off-policy distribution estimation in reinforcement learning. %
We validate the efficacy of our method empirically on synthetic data as well as on a slate recommendation simulator constructed from real-world data (MovieLens-20M).
Our results show a significant reduction in estimation variance and improved sample efficiency over prior work across a
range of slate structures. %

\end{abstract}

\section{Introduction}
\label{sec:intro}

Recommendation services are ubiquitous throughout industry
\cite{bobadilla2013recommender,lu2015recommender}.
A common variant of recommendation consists of suggesting multiple items to a user simultaneously, often termed recommendation \textit{slates}, where each position, (a.k.a., a  \textit{slot}), can take multiple possible values \cite{sarwar2000analysis}. For example,
webpage layouts of news or streaming services have separate slots for each category of content where each slot can display any of the items from the category for that slot.
The items are suggested to the user based on a recommendation strategy, called a \textit{policy} and the user response is encoded into a \textit{reward} \cite{li2010contextual}.
A crucial problem for selection and improvement of recommendation strategies is to evaluate the efficacy of a slate policy by estimating the expected reward of that policy. 

One of the simplest and most effective approaches to policy evaluation, often employed in industrial settings, is A/B testing~ \cite{gomez2015netflix,kohavi2017online,feitelson2013development}. This involves randomly assigning users to receive the item recommended by one of two candidate policies, and the relative performance of each policy is directly measured.
However, A/B testing involves deploying the new policy online, which may be infeasible in many settings due to practical or ethical considerations.
As a result,
it is often necessary to employ offline \emph{off-policy evaluation}, in which interaction data collected (offline) from previously deployed policies (off-policy) is used to estimate statistics of the expected performance, risk and other metrics of interest for a new target policy, without actually deploying it online.
This ensures that policies with undesirable outcomes are not deployed online.

A large amount of literature addresses the problem of off-policy evaluation in the non-slate setting~\cite{dudik2014doubly,wang2017optimal,thomas2015high}, where the majority of methods rely on some version of importance sampling \cite{horvitz1952generalization}.
Applied to the slate setting, these methods result in very large importance weights that result in high variance estimates due to the combinatorially large action space on which slate policies operate.
Recent work addresses this deficiency, by introducing methods that leverage the structure in slate actions and rewards to address the high variance.
\citet{swaminathan2017off,vlassis2021control} leverage reward additivity across slot actions to propose estimators for estimation of the expected reward of target slate policies.
\citet{mcinerney2020counterfactual} propose a Markov structure to the slate rewards for sequential interactions.

All the aforementioned methods provide solutions for estimating the \textit{expected reward} for a target policy.
However, in many scenarios where recommendation systems are used, such as those with large financial stakes and in healthcare applications, practitioners are concerned with evaluation metrics such as the behavior of the policy at extreme quantiles and the expected performance at a given risk tolerance~(CVaR).
These quantities need the full \textit{reward distribution} for estimation, which renders prior work which estimates the expected reward inapplicable.
A notable exception is \citet{chandak2021universal} who provide a method for off-policy estimation of the target policy's cumulative reward distribution using ideas similar to importance sampling, allowing for the computation of various metrics of interest, but the estimator is intractable in the slate setting. %

In this work, we propose \underline{s}late \underline{un}iversal \underline{o}ff-policy evaluation~(SUnO), a method that allows for off-policy estimation of target reward distribution (Theorem \ref{thm:main_result}) for slate recommendation policies.
SUnO applies the core ideas from the universal off-policy evaluation (UnO) method~\cite{chandak2021universal} to the slate setting, leveraging an additive decomposition of the conditional reward distribution. This makes it possible to perform off-policy estimation in structured high dimensional action spaces without incurring prohibitively high estimation variance.
We highlight how the estimator can readily be adapted to other generalized decompositions of reward while continuing to be unbiased.
Finally, we provide an empirical evaluation comparing against UnO, where the proposed estimator shows
significant variance reduction, improved sample efficiency, and robust performance even when the conditions for unbiasedness of the estimator are not met.

The main contributions of our work are:
\begin{itemize}
    \item We propose an unbiased estimator for the off-policy reward distribution for slate recommendations under an additively decomposable reward distribution, generalizing prior results for slate off-policy evaluation to the distributional setting.
    \item We theoretically demonstrate how the estimator readily generalizes to slate rewards that do not decompose additively over slots. %
    \item We empirically demonstrate the efficacy of the proposed estimator on slate simulators using synthetic as well as real world data, on a range of slate reward structures.
\end{itemize}

\section{Background and Notation}
\label{sec:notation}

We first formulate the slate recommendation system as a contextual bandit with a combinatorial action space. Each slate action
has $K$ dimensions where each dimension is a slot-level action.
The user-bandit interaction results in a random tuple $(X, A, R)$ at each step, where $X \sim d_X(\cdot)$ is the user context,
$A$ is the slate action generated by the recommendation strategy where $A = [A^k]_{k=1}^K$ is composed of $K$ slot-level actions, and $R \sim d_R(\cdot \mid A, X)$ is the scalar slate-level reward. Since the rewards are observed only at the slate level and not at a per-slot level, we use reward and slate reward interchangeably. {Each} slot-level action can take upto $N$ candidate values, leading to a combinatorially large action space of the order ${N \choose K}$.

A \textit{logging} policy
$\mu(A \mid X) = \Pr(A \mid X)$ recommends slate actions conditioned on user context $X$ is deployed online to collect a dataset for offline evaluation. The offline dataset consists of $n$ i.i.d. samples $D_n = \{(X_i, A_i, R_i)\}_{i=1}^n$, generated by the user-bandit interaction. We focus on the case where $\mu$ is a factored policy, that is,
\[ \mu(A \mid X) = \prod_{k=1}^K \mu_k \big(A^k \mid X \big)\]
where $K$ is the number of slots. Data collection with factored uniform logging policies is standard in practice \cite{swaminathan2017off}.
\textit{Off-policy evaluation} is the task of utilizing data $D_n$ logged using a policy $\mu$, to evaluate a target policy $\pi$ by computing evaluation metrics from the target reward under $\pi$. Standard methods focus on the estimation of the expected reward under the target policy.
In this work, our focus is on the estimation of quantities that go beyond just the expected target reward by estimating the reward \textit{distribution}.

Throughout the paper, the sample estimates of any quantity $y$ are denoted by $\hat{y}_n$ where the subscript indicates the number $(n)$ of data points used for estimation.
For instance, the cumulative reward distribution observed for a policy $\pi$ is denoted by $F^\pi(\nu)$. The sample estimate of the distribution will be denoted by $\hat{F}^{\pi}_n(\nu)$.

\subsection{Related Work on Off-Policy Evaluation}
\label{sec:ope_bg}
Importance sampling (IS) \cite{horvitz1952generalization,sutton2018reinforcement}, also known as inverse propensity scoring (IPS) \cite{dudik2014doubly}, provides a technique for unbiased estimation of expected target reward. However, it suffers from high variance in large action spaces.
There are numerous extensions to IS for variance reduction \cite{dudik2011doubly,thomas2015safe,kallus2019intrinsically}.
The IS estimator and all methods derived from it rely on a standard common-support assumption. We use a weaker form of this assumption that requires support at the slot level instead of the entire slate.
\begin{ass}[Common Support]
    The set $D_n$ contains i.i.d. tuples generated using $\mu$, such that for some (unknown) $\varepsilon>0, \mu_k(A^k \mid X)<\varepsilon \Longrightarrow$ $\pi(A^k \mid X)=0;$ $\forall ~k, X, A$.
    \label{ass:support}
\end{ass}
Applying the methods to slates, \cite{swaminathan2017off} assume additivity of slate rewards as a way to reduce the estimation variance of the expected reward of the target policy.
Further variance reduction is obtained by using control variates \cite{vlassis2021control}.
\citet{mcinerney2020counterfactual} assume a Markov structure to the slate rewards during sequential item interaction and propose an estimator for the expected target reward. We refer the reader to these papers for additional references on slate recommendations.

These off-policy estimators, along with most others in the literature, provide estimates of the \textit{expected} reward of the target policy \cite{li2018offline}. However, the expected value is usually not sufficient for comprehensive off-policy analysis, particularly in the case of risk assessment that is crucial for recommendation systems \cite{shani2011evaluating}. Additional metrics of interest, often those computable from the whole \textit{reward distribution} are necessary in practice \cite{keramati2020being,altschuler2019best}.
For example, metrics like value at risk are used risk analysis of a new recommendation strategy.
To that end, work on universal off-policy estimation (UnO) \cite{chandak2021universal} uses ideas motivated by importance sampling to estimate the whole cumulative distribution of the reward under the target policy.
However, in combinatorially large action spaces, as with the slate problems we consider here, the UnO estimator can incur prohibitive variance.
Our proposed estimator utilizes possible structure in slate rewards to circumvent this issue.

\subsection{Structure in Slate Rewards}

The combinatorial action space of slates becomes a key challenge for most general methods like importance sampling (IS) for off-policy evaluation \cite{wang2017optimal}.
The generality of the approach results in them not fully leveraging the structure present in slate rewards \cite{sunehag2015deep} and thus general IS-based approaches frequently suffer from high variance.
Prior work leverages user behavior patterns while interacting with slates, which are encoded into structured rewards (e.g., time spent, items clicked, etc.).
Some examples of structure in slate rewards are Markov structure for observed slot-level rewards \cite{mcinerney2020counterfactual}, the dependence of slate reward only on the selected slot \cite{ie2019slateq}, and unobserved slot-level rewards with an additively decomposable slate reward \cite{swaminathan2017off,vlassis2021control}.

The last one is of particular interest, where the \textit{additivity of expected reward} \cite{cesa2012combinatorial} posits that the conditional mean slate-level reward decomposes additively as the sum of (arbitrary) slot-level latent functions, i.e., $\E[R \mid A, X] = \sum_{k=1}^K \phi_k(A_k, X)$. This has been leveraged to obtain a significant reduction in estimation variance for off-policy evaluation \cite{swaminathan2017off,vlassis2021control}.
This decomposition captures the individual effects of each slot. It may readily be generalized to capture non-additive joint effects of more than one slot action, for example, to capture the effects of \textit{pairs} of slot-actions, one may consider the decomposition:
\begin{equation}
    \E[R \mid A, X] = \sum_{k=1}^{K} \sum_{j=k}^{K} \phi_{jk}(A_k, A_j, X) \, .
    \label{eq:general_rew}
\end{equation}
It may further be generalized to capture the combined effects of $m$-slots.
Note that in the most general case, for $m=K$, the reward does not permit any decomposition over slots.

Analogous to the above structural conditions, we posit a condition that allows us to perform \textit{consistent and unbiased} estimation of the target off-policy \textit{distribution}.
\begin{ass}[Additive CDF]
    There exists an additive decomposition of the conditional cumulative density function (CDF) of the slate reward as the sum of (arbitrary) slot-level latent functions:
    \[F_R(\nu) = \sum_{k=1}^{K} \psi_{k}(A_{k}, X, \nu) ,~\forall \nu \]
    where $F_R(\nu) := \Pr(R \leq \nu \mid A, X) $ .
    \label{ass:additive_cdf}
\end{ass}
The slot-level rewards, if any, are unobserved. The condition just assumes that an additive decomposition \textit{exists} and does not require knowledge of the constituent slot-level functions.
We demonstrate empirically
that this condition is often a close approximation for real-world data and that our estimator performs better than more general methods even when this condition happens to be an inexact approximation, i.e., when a perfect decomposition does not exist.

This decomposition may also be generalized to capture the combined effects of $m$-slots.
In line with prior work \cite{wen2015efficient,kale2010non,kveton2015tight,swaminathan2017off,vlassis2021control,ie2019slateq}, we focus our
analysis for the case $m=1$, which proves to be effective in practice as corroborated by empirical analysis.
We additionally provide derivations for how the estimator can readily generalize to cases where $m>1$, along with a theoretical analysis of its properties.

It is worth noting that an additively decomposable reward CDF always implies an additive expected reward %
by definition
and the former often serves to be a close approximation when the latter holds.
This is helpful since a commonly used metric for evaluating the performance of slates, the normalized discounted cumulative gain (nDCG) \cite{burges2005learning}, is an additively decomposable metric, and it has been used in the past for defining the slate reward \cite{jarvelin2017ir,swaminathan2017off}.

\section{Slate Universal Off-Policy Evaluation}
\label{sec:suno}

We will now turn to off-policy evaluation in slates as an off-policy reward distribution estimation task.
The core idea builds upon the framework of \citet{chandak2021universal} who use importance weights $\rho$
in an estimator of the reward CDF of a target policy from logged data $D_n \sim \mu$.
In the case of slates, the the importance weight $\rho$ comprises of probability ratios over all slot actions.
The most direct approach for defining $\rho$ in the case of a factored logging policy $\mu$ is to consider a formulation analogous to importance sampling by taking the \textbf{product} of the slot-level probabilities (Equation \eqref{eq:iw_G}).
This approach will be plagued by high variance when the size of the slate $K$ is large.
To remedy this, our proposed algorithm SUnO utilizes the structure in slates provided by Assumption~\ref{ass:additive_cdf} wherein the CDF of the slate level reward admits an additive decomposition. In place of $\rho$, we define an importance weight $G$ (Equation \eqref{eq:iw_G}) that is a \textbf{sum} of slot-density ratios.
\begin{equation}
    \begin{split}
        &\rho = \frac{\pi(A|X)}{\mu(A|X)} = \prod_{k=1}^K \frac{\pi\big(A^k \mid X\big)}{\mu_k(A^k \mid X)} \\
        &G = \sum_{k=1}^K \left( \frac{\pi\big(A^k \mid X\big)}{\mu_k\big(A^k \mid X\big)} -1 \right) + 1 \,
    \end{split}
    \label{eq:iw_G}
\end{equation}

\begin{algorithm*}[tb]
    \caption{SUnO$(\nu)$}
    \begin{flushleft}
        \textbf{Input}: $\pi$, $\mu$, $\nu$, $\{(X_i, A_i, R_i)\}_{i=1}^n \sim D_n$\\
        \textbf{Output}: $\hat{F}_n^\pi(\nu)$
    \end{flushleft}
    \begin{algorithmic}[1] %
        \STATE $s_\nu = 0$ \hfill \COMMENT{Initialize counter and iterate over the logged dataset}
        \FOR{$i = 1,2, \dots, n$}
        \STATE $G_i \gets 1 - K + \sum_{k=1}^K \frac{\pi\left(A_i^k \mid X_i\right)}{\mu_k\left(A_i^k \mid X_i\right)} $ \hfill \COMMENT{Compute the importance weight (Equation \eqref{eq:iw_G})}
        \STATE $s_\nu \gets s_\nu + \I\{R_i \leq \nu\}G_i$ \hfill \COMMENT{Add to counter if reward is less than $\nu$}
        \ENDFOR
        \STATE \textbf{return} $\hat{F}_n^\pi(\nu) = s_\nu/n$ \hfill \COMMENT{Normalize and return counter}
    \end{algorithmic}
    \label{alg:suno}
\end{algorithm*}
The estimator for the target distribution counts the number of samples for which the reward is less than a threshold $\nu$ and reweighs that count with importance weights to be reflective of the counts under the target policy $\pi$. In expectation, the counts reflect the probabilities of obtaining a reward less than $\nu$, providing the value of the reward CDF at $\nu$, denoted by $F^\pi(\nu)$.
Use of the importance weight $G$ in place of $\rho$
results in significantly lower variance in estimation and improved effective sample size, while keeping the estimator unbiased.
It is easy to confirm that under a factored logging policy, $\E_\mu[G] = 1$.
Below we prove that this importance weight
allows for a change of distribution of the expected slate reward
and can be used for estimating the target reward CDF. Our main result is the following:

\begin{restatable}{theorem}{mainresult}
    Let $R$ be a real-valued random variable that denotes the slate reward and admits an additive
    decomposition of its conditional cumulative distribution $F_R(\nu)$ (Assumption \ref{ass:additive_cdf}).
    Under a factored $\mu$ and Assumption \ref{ass:additive_cdf}, $$F^\pi(\nu) = \E_\mu [G \cdot \I\{R \leq \nu \}], ~\forall ~\nu$$
    \label{thm:main_result}
\end{restatable}
Thus, a weighted expectation of the indicator function, with weights given by $G$, gives the target CDF. Based on this result, we propose the following sample estimator for $F^\pi(\nu)$ that uses data $D_n \sim \mu$,
\[\hat{F}^\pi_n(\nu) := \frac{1}{n} \sum_{i=1}^n G_i \, \I\{R_i \leq \nu\}, ~\forall ~\nu \]

This estimator, called \textit{slate universal off-policy estimator} (SUnO) is outlined in Algorithm \ref{alg:suno}. In the following result, we establish that SUnO leverages the additive structure in Assumption \ref{ass:additive_cdf} to obtain an unbiased and pointwise consistent estimate of the CDF of the target policy.
The proofs of both results may be found in the Appendix.
\begin{restatable}{theorem}{unbiasedestimator}
    Under Assumption \ref{ass:additive_cdf}, $\hat{F}^\pi_n(\nu)$ is an unbiased and pointwise consistent estimator of $F^\pi(\nu)$.
    \label{thm:unbiased_consistent}
\end{restatable}

It is important to note that
analogous to \citet{swaminathan2017off}
our estimator does \emph{not} require knowledge of the specific functions ($\psi_k$'s) in the decomposition of the conditional CDF in Assumption \ref{ass:additive_cdf}; it only assumes the \emph{existence} of a set of such latent functions, and a corresponding additive decomposition of the conditional CDF, to attain unbiased estimation.
Even in cases where the assumption is not satisfied, our method (Algorithm \ref{alg:suno}) performs robustly, as we demonstrate in our experiments.

The estimated target CDF can be used to compute metrics of interest as functions of the CDF (for example, mean, variance, VaR, CVaR, etc.).
Some of these metrics are non-linear functions of the CDF (VaR, CVaR) and thus their sample estimates would be biased estimators.
This is to be expected \cite{chandak2021universal}.
Metrics that are linear functions of the CDF however have unbiased sample estimators. Thus, an unbiased target CDF estimator serves to be a ``one-shot'' solution for most metrics of interest, though unbiasedness holds only for certain metrics. We demonstrate the estimation of some of these metrics in our empirical analysis.

\subsection{Properties of SUnO}

\subsubsection{Variance:}
The estimator enjoys significantly low variance for target estimation. The key factor is that the estimator uses importance weights that are a \textit{sum} of slot level density ratios as opposed to a \textit{product} as is used UnO \cite{chandak2019learning}. Particularly in the slate setting, the latter methods suffer from enormous variance and reduced effective sample size, as we demonstrate empirically.

Consider the worst-case variance of the two estimators. From Assumption \ref{ass:support} we have $0 \leq \frac{\pi(A^k|X)}{\mu_k(A^k|X)} \leq \frac{1}{\epsilon}$, which implies
\begin{align*}
    \Var(\text{UnO}) = O\left(\frac{1}{\epsilon^K} \right); \quad \Var(\text{SUnO}) = O\left( \frac{K}{\epsilon} \right)
\end{align*}
Thus the worst-case variance of SUnO grows linearly with an increase in the size of the slate ($K$) while that of UnO grows exponentially.

\subsubsection{Generalization to $m$-slot reward decomposition}

\begin{figure*}[htb]
    \centering
    \includegraphics[width=0.9\textwidth]{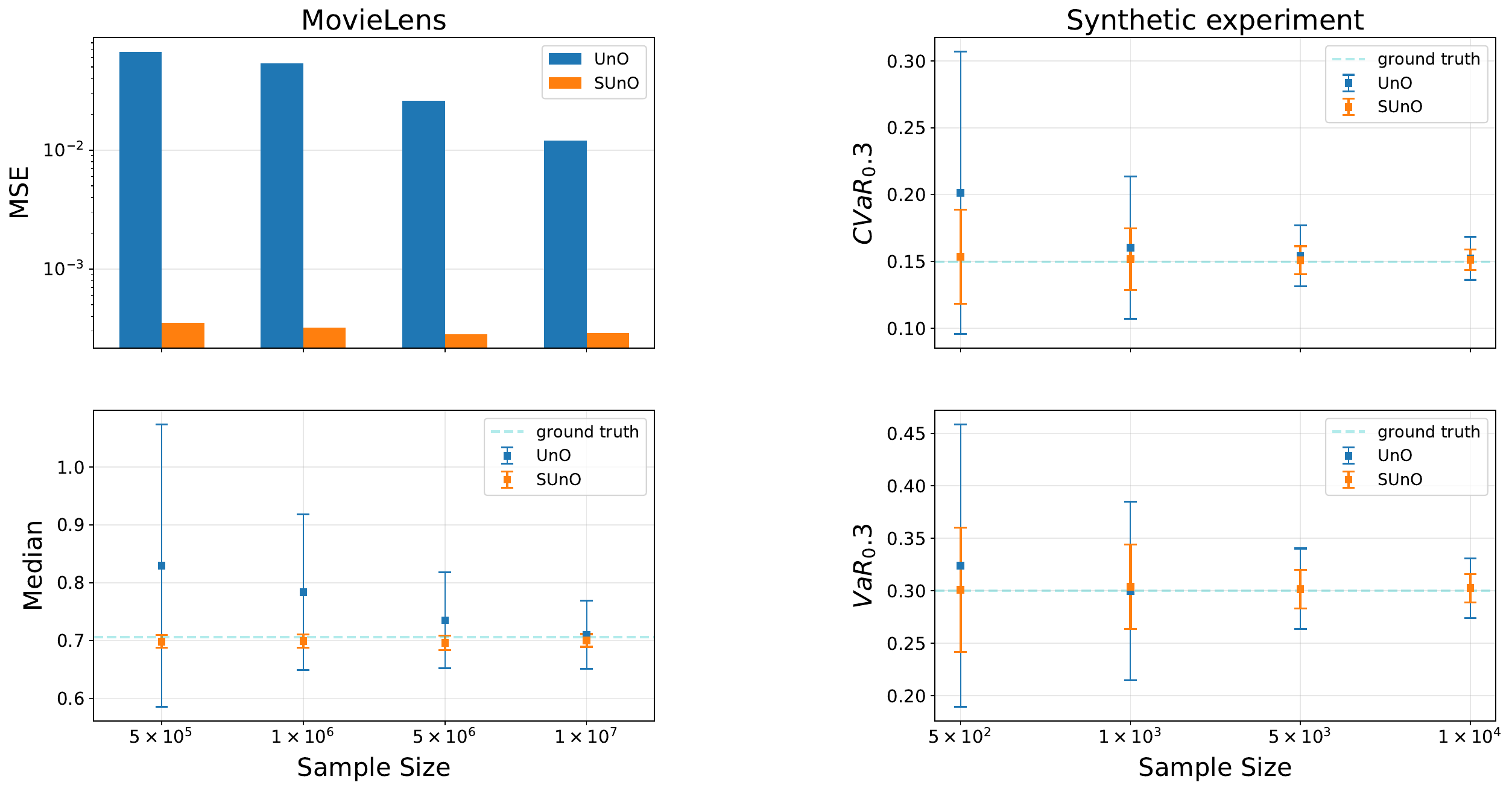}
    \caption{(left) Movielens-20M: (top) MSE for mean computed from the estimated target CDF for increasing sample sizes. SUnO performs significantly better in terms of bias and variance compared to UnO. The same follows for median estimation (bottom) where it demonstrates much better sample efficiency and lower estimation variance as seen by the error bars.
        (right) Synthetic Experiment: Estimates of CVaR$_{0.3}$ (top) and VaR$_{0.3}$ (bottom) computed from the estimated target CDF. In this setting where Assumption \ref{ass:additive_cdf} is satisfied, SUnO performs better than UnO in terms of estimation variance, sample efficiency, and estimation accuracy as expected.}
    \label{fig:synth_and_movielens}
\end{figure*}

As noted earlier, the structural assumptions in slate rewards may be generalized to account for the joint effects of multiple slots. The proposed estimator readily applies to such generalizations.
For instance, consider the case where the conditional reward CDF decomposes into terms composed of $m$ slot-actions.
\begin{equation}
    F_R(\nu) := \sum_{1 \leq k_1 < k_2, \dots < k_m \leq K} \psi_{k_{1:m}}(A^{k_{1:m}}, X, \nu)
    \label{eq:m_way_cdf}
\end{equation}
where $k_{1:m}$ is used as shorthand to denote the indices $k_1, k_2, \dots, k_m$. It can be seen that this decomposition consists of ${K \choose m}$ terms. The importance weight from Equation \ref{eq:iw_G} with a minor change provides an \textit{unbiased} and \textit{consistent} off-policy estimator for this reward structure.
Define
\begin{equation}
    G_m = \sum_{1 \leq k_1 < k_2, \dots < k_m \leq K} \left(\prod_{i=1}^{m}\frac{\pi(A^{k_i}|X)}{\mu_{k_i}(A^{k_i}|X)} - 1\right) + 1
    \label{eq:m_way_G}
\end{equation}
to be the importance weight. With a derivation similar to Theorem \ref{thm:main_result} we show the following result,
\begin{restatable}{cor}{mwaygeneralization}
    Let $R$ be a random variable that denotes the slate reward, and permits a decomposition into (latent) functions of $m$ slots as in Equation \ref{eq:m_way_cdf}. Under a factored $\mu$, we have
    \[F^\pi(\nu) = \E_\mu [G_m \cdot \I\{R \leq \nu \}], ~\forall ~\nu\]
\end{restatable}

The proof of the result may be found in the Appendix. The derivation highlights how the result may be extended to other forms of decomposition of the slate reward.
Note that when $m=K$, the importance weight reduces to $G = \prod_{i=1}^K Y_{k+i}$ which is the same as $\rho$ as used in UnO. This is because at $m=K$ the reward permits no decomposition. The proposed estimator may then be interpreted as a generalization of UnO to various reward decompositions, where UnO forms a special case wherein the reward does not permit any decomposition.
We now empirically study the case where $m=1$, in line with prior work. The reward decomposition at $m=1$ proves to be a sufficiently accurate approximation to real-world data as we demonstrate in the empirical section that follows.

\section{Empirical Analysis}
\label{sec:experiments}

\begin{table*}[t]
    \centering
    \parbox{.45\linewidth}{
        \caption*{(a) Synthetic experiment}
        \begin{tabular}{|c|c|c|c|c|}
            \hline
            Sample size & 0.5 $\times 10^3$ & 1 $\times 10^3$ & 5 $\times 10^3$ & 10 $\times 10^3$ \\
            \hline
            SUnO        & \textbf{0.131}    & \textbf{0.102}  & \textbf{0.059}  & \textbf{0.049}   \\
            UnO         & 0.256             & 0.191           & 0.098           & 0.077            \\
            \hline
        \end{tabular}
    }
    \qquad
    \parbox{.45\linewidth}{
        \centering
        \caption*{(b) MovieLens}
        \begin{tabular}{|c|c|c|c|c|}
            \hline
            Sample size & 0.5 $\times 10^6$ & 1 $\times 10^6$ & 5 $\times 10^6$  & 10 $\times 10^6$ \\
            \hline
            SUnO        & \textbf{0.169}    & \textbf{0.173}  & \textbf{ 0.181 } & \textbf{0.184}   \\
            UnO         & 0.441             & 0.410           & 0.31912          & 0.252            \\
            \hline
        \end{tabular}
    }
    \caption{The tables report the average Kolmogorov-Smirnov statistic for the estimated CDFs for (a) Synthetic experiment and
        (b) MovieLens. The results demonstrate that SUnO estimates the target CDF more accurately with better sample efficiency.}
    \label{table:KSS}
\end{table*}
We investigate the following questions in the empirical analysis:
\\\textbf{RQ1}: Does the estimator have low estimation variance and high sample efficiency when the slate rewards have an additive structure?
\\\textbf{RQ2}: Does the method accurately estimate the off-policy distribution, and metrics from it?
\\\textbf{RQ3}: Does the method apply to settings where conditions for unbiasedness of the estimator do not hold?

To that end, we evaluate the performance of SUnO in a range of slate recommendation settings, comparing its performance against UnO \cite{chandak2021universal}, a general estimator that does not make any structural assumptions.
Better performance is defined as having lower estimation error and variance, along with improvement in sample efficiency.

We begin by evaluating the estimators on synthetic data that follows the additive CDF structure to corroborate theoretical results (RQ1).
We then proceed to real-world data experiments and use the additively decomposable metric nDCG \cite{swaminathan2017off} as the slate reward (RQ2, RQ3).
We test our estimator on a publicly available dataset - MovieLens-20M \cite{harper2015movielens} - and on a semi-synthetic slate simulator - Open Bandit Pipeline \cite{saito2020open}.

We develop a procedure to \textit{construct a slate simulator} from ratings datasets like MovieLens.
We evaluate SUnO in settings where
the slate reward does not by construction satisfy either the additive CDF or the additive reward conditions,
and see it demonstrate robust performance.

\noindent\textbf{Implementational note:}
Algorithm \ref{alg:suno} outlines the steps for estimating the target CDF at any reward value $\nu$.
The reward, in general, takes on continuous real values and implementationally it is not practical to estimate the target CDF at all continuous values of $\nu$. In practice, an empirical estimate of the CDF may be computed at discrete points over the range of rewards. In between those points, the value of the CDF is kept constant.
Consequently, we compute the target CDF at evenly spaced points over the range of rewards for both estimators in the experiments that follow. The granularity of this discretization of the domain of the CDF reflects in the granularity of the estimated CDF. To ensure accuracy and relative smoothness in the estimated CDF, we choose a very fine level of discretization relative to the range of reward for each experiment.

\noindent\textbf{Metrics:} We define a few metrics that we utilize in our empirical analysis.
Value at Risk (VaR$_{\alpha}$) \cite{rockafellar2000optimization,wirch2001distortion} denotes the value of the reward such that the probability of observing rewards greater than that value is $1-\alpha$. Correspondingly, Conditional Value at Risk (CVaR$_{\alpha}$) denotes the expected value of the rewards given that that observed reward is less than VaR$_{\alpha}$. These metrics are used for risk assessment of policies and we compute them from the estimated target reward distribution. To evaluate the estimation error in our estimated target distribution, we use the Kolmogorov-Smirnov statistic \cite{stephens1974edf} which measures how well the estimated target distribution matches ground-truth distribution by assessing the largest discrepancy between their cumulative probabilities.

All the experiments have a factored uniform-random logging policy. The error bars denote one standard error. The code is available at: 
\url{https://github.com/shreyasc-13/suno}.

\subsection{Synthetic Experiments}
\label{sec:synth_exp}
We begin by synthetically generating data where the slate reward permits the additive CDF structure (Assumption \ref{ass:additive_cdf}).
\\ \textbf{Setup}:
We consider the non-contextual bandit setting for ease of analysis, and the same may easily be extended to a contextual setting.
To construct the data-generating reward distribution, each $\psi_k$ is set to a monotonic non-decreasing function by assigning slices of a sigmoid function to the corresponding $\psi_k$'s for each $(A^k)$.
This manner of construction ensures that the resultant sum of the functions, the CDF, is again a monotonic non-decreasing function. The $\psi_k$'s are appropriately normalized.
For these experiments, we set the number of slots $K = 3$ and the number of actions in each slot to $N=3$.
The target policy is a deterministic policy that chooses one action per slot, where the action for each slot is assigned randomly at the start of the experiment and held constant for all experiments.
\\
\textbf{Experiment}: We compare the performance of the estimators on two fronts:
\begin{enumerate}[leftmargin=*]
    \item \textit{Goodness-of-fit of CDF}: We report the average Kolmogorov-Smirnov statistic of the estimated target CDFs. The ground truth target CDF is computed by executing the target policy on the simulator.
    \item \textit{Tail measures}: We compute the CVaR$_{0.3}$ and VaR$_{0.3}$ from the target CDFs estimated by the two estimators.
\end{enumerate}
The experiments are run for different logged data sizes and the results are averaged over 1000 trials.
\\\textbf{Results} [Table \ref{table:KSS}, Figure \ref{fig:synth_and_movielens}]: Since SUnO leverages the additive structure in rewards, it estimates the CDF and the tail measures with lower variance and estimation error, while being more sample efficient.
If a single slot action has a zero probability of occurring under the target policy, the importance weight used by UnO for the entire slate goes to zero since it comprises of the product of slot-level density ratios. This is not the case for SUnO which is thus able to utilize a larger effective sample size and be more sample efficient.
Note that with an increase in sample size, both estimators tend to the ground-truth values as they are consistent and unbiased. SUnO has a significantly lower variance in estimation as seen by the error bars in Figure \ref{fig:synth_and_movielens}.

\subsection{Real-World Data}
\label{sec:real_data}
In this section, we first introduce a procedure for converting a recommender system ratings dataset (like MovieLens) to a slate recommendation simulator with additive rewards and proceed to evaluate our method on the simulator. The procedure for constructing the simulator follows.
\subsubsection{Simulation Setup}
\begin{enumerate}[leftmargin=*]
    \item Learn a user-item preference matrix $B$ along with the user context embedding $X$. For $m$ users and $l$ item, $B \in \mathbb{R}^{m \times l}$ and $X \in \{0,1\}^{m}$. We follow the steps outlined in \cite{steck2019embarrassingly} to learn $B$ from rating data. $X$ is a binary vector that encodes user-item interaction history. An alternative method for learning embeddings could be \cite{elahi2020learning}.
    \item To limit our setup to approximately 10k unique users, we trim the set of users to those that have an interaction history of 10 to 15 items.
    \item Compute the \textit{ground truth} preference scores for each user by computing the product of a user's context embedding with the preference matrix (\texttt{$x \cdot B$}).
    \item To make the simulator tractable, we trim the action set by retaining the top 20 preferred actions per user based on each user's ground truth scores ($N=20$).
    \item For a slate action $A$, a ranking metric like nDCG can be set as the slate reward $R$ and works well in practice \cite{vlassis2021control,swaminathan2017off}.
\end{enumerate}
\noindent
\textbf{Experiment}:
First, we set up a slate simulator as described above using the MovieLens-20M dataset to estimate $B$ and $X$. A uniform random factored logging policy is used for creating the offline dataset for evaluating the estimators.
We consider an $\epsilon$-greedy target policy. For each user, it picks the top $K$ preferred actions (one per slot) with probability $1 - N\epsilon$ and a uniform random action from the user's action set with probability $\epsilon$.
Here $N=20, K=5, ~\epsilon=0.01$ and results are averaged over 50 trials. We analyze:
\begin{enumerate}[leftmargin=*]
    \item \textit{Goodness-of-fit of CDF}: We report the average Kolmogorov-Smirnov statistic of the estimated CDFs against the ground truth CDF.
          The ground truth CDF is computed by executing the target policy on the simulator.
    \item \textit{Metrics computed from the CDF}: We compute the mean and 0.5-quantile (median) from the estimated CDF.
\end{enumerate}
\textbf{Results} [Table \ref{table:KSS}, Figure \ref{fig:synth_and_movielens}]: The experiments demonstrate that although only the additive reward condition is met and not Assumption \ref{ass:additive_cdf} ,
SUnO estimates the target CDF with fewer samples (Table \ref{table:KSS}) than UnO.
Our estimator has a significantly lower estimation variance for metrics computed from the CDF, as seen by the error bars for median computation and the mean squared error for the expected value computation in Figure \ref{fig:synth_and_movielens}.
Note that the mean squared error (MSE) captures both the bias and variance in estimation.

\subsection{Non-Additive Reward Structure}
\label{sec:ass_violation}

\begin{table}
    \centering
    \begin{tabular}{|c|c|c|c|}
        \hline
        Sample size & 0.5 $\times 10^5$ & 1 $\times 10^6$ & 5 $\times 10^6$ \\
        \hline
        SUnO        & \textbf{0.253}    & \textbf{0.257}  & \textbf{0.269}  \\
        UnO         & 0.543             & 0.541           & 0.567           \\
        \hline
    \end{tabular}
    \vspace{0.2em}
    \caption{The table reports the mean squared error for mean computation from the estimated CDFs. Even in settings where the slate reward is not additive, our method continues to perform better than the structure-agnostic estimator.}
    \label{fig:obp}
\end{table}
Finally, we evaluate the estimators in a setting where neither the additive reward nor the additive CDF conditions are satisfied.\\
\textbf{Simulator:} We use the Open Bandit Pipeline (OBP) slate bandit simulator \cite{saito2020open} that uses the synthetic slate reward model described in \cite{kiyohara2022doubly}.
It models higher-order interactions among slot actions and thus does not trivially satisfy Assumption \ref{ass:additive_cdf} or the additive slate reward structure.
We use the \textit{cascade additive reward model} defined in OBP for these experiments.\\
\textbf{Experiment}: Similar to the MovieLens experiments, we observe the estimation error for the target mean computed from the estimated CDF. A uniform random logging policy is used to generate the offline dataset and the target policy defined here\footnote{\url{https://github.com/st-tech/zr-obp/blob/master/examples/quickstart/synthetic_slate.ipynb}} is evaluated. We set $K=3, N=10$, and the results are averaged over 10 trials.\\
\textbf{Results} [Table \ref{fig:obp}]:
In this setting, we cannot expect unbiased estimates of the mean from SUnO
since the additive CDF condition is required for unbiased estimation of the target CDF.
Nonetheless,
SUnO continues to perform significantly better in terms of the mean squared error for the mean estimation compared to UnO, which does not make any structural assumptions and is an unbiased estimator in this setting. Here $K$ is set to a relatively small value and a large gap in performance between the two estimators can be expected larger $K$.

\section{Discussion and Conclusion}
We proposed an estimator (SUnO) for off-policy estimation of the target \textit{reward distribution} in slate recommendations modeled as a bandit problem.
Under an additively decomposable conditional CDF, the estimator is unbiased and consistent.
The proposed estimator leads to significant reduction in estimation variance and an increase in effective sample size as compared to the estimator of \citet{chandak2021universal} for the slate setting. We demonstrate estimation gains on synthetic as well as real-world data experiments.
The estimator also readily extends to other reward decompositions that capture the joint effects of slot actions.

In future work, variance reduction techniques can be applied for further variance gains. For instance, %
one may consider a self-normalized version of SUnO that incurs some bias but provides further variance reduction akin to weighted importance sampling \cite{koller2009probabilistic}.
Control variates can also be adopted for variance reduction. For example, by recapitulating the analysis of \cite{vlassis2021control}, one can derive an optimal control variate ($w^*$) for SUnO at each $\nu$.
Further analysis and experiments with such methods are left to future work.

One must note while SUnO provides unbiased estimates for the target CDF and its linear functions under our assumptions, many metrics of interest, such as variance or CVaR,  are not linear functions of the CDF. Another direction for future work would be to extend our results to the unbiased estimation of risk functionals like CVaR \cite{huang2021off2}. %
Finally, it would be interesting to develop techniques for discovering the decomposition structure in rewards so that the appropriate unbiased off-policy estimators can be used.

\section{Acknowledgements}
The research was supported by and partially conducted at Adobe Research. We are also immensely grateful to the four anonymous reviewers who shared their insights and feedback.

\bibliography{references}

\appendix
\newpage
\onecolumn
\appendix

\label{sec:appendix}

\section{Proofs of Main Results}
\label{sec:proofs}

\mainresult*
\begin{proof}[Proof for Theorem \ref{thm:main_result}]
    Under Assumption \ref{ass:additive_cdf} we have,
    $F_R(\nu) = \E[\I\{R \leq \nu\}|A,X] = \sum_{k=1}^K \psi_k(A_k, X, v) ,~\forall \nu$.
    For ease of notation, let $Y_k = \frac{\pi(A^k|X)}{\mu_k(A^k|X)}$.
    Following the proof structure from \cite{vlassis2021control},
    \[
        \begin{split}
            &E[G \cdot \I\{R \leq \nu\}|A,X] = \big(1 - K + \sum_{k=1}^K Y_k \big) \big(\sum_{k=1}^K \psi_k(A^k, X) \big) \\
            &= \sum_{k=1}^K Y_k \psi_k(A^k, X, \nu) + \sum_{k=1}^K (1 - K + \sum_{j\neq k}^K Y_j ) \psi_k(A^k, X, \nu)
        \end{split}
    \]
    Take an expectation over $A \sim \mu(.|X)$.
    It can be seen that the second term equals 0. Due to importance sampling, the first term equals $\E_\pi[\I\{R \leq \nu\} | X]$.
    Taking an expectation over $X \sim d_X(.)$, by the total law of expectation,
    \[\E_\mu[G \cdot \I\{R \leq \nu\}] =  \E_\pi[\I\{R \leq \nu\} ] = F^\pi(\nu)\]
\end{proof}

\unbiasedestimator*
\begin{proof}[Proof for Theorem \ref{thm:unbiased_consistent}]
    It can be shown that $\hat{F}^\pi_n(\nu)$ is an unbiased estimator of $F^\pi(\nu)$ by taking an expectation of $\hat{F}^\pi_n(\nu)$ over datasets $D \sim \mu$,
    where (a) follows from Theorem \ref{thm:main_result}.
    \[\E_{D\sim \mu}\big[\frac{1}{n} \sum_{i=1}^n G_i \I\{R_i \leq \nu\}\big] \stackrel{(a)}{=} \frac{1}{n} \sum_{i=1}^n \E_{\pi}\big[\I\{R_i \leq \nu\}\big] = F^\pi(\nu)\]

    To establish almost sure convergence of the estimator, note that
    each data point in $D_n$ is i.i.d. Additionally, the magnitude of $G_i$ is bounded under the assumption of common support, since each slot-density ratio can at most be $1/\epsilon$. As a result, the variance is $M_i :=  G_i \I\{R_i \leq \nu\}$ is bounded. Thus $M_i$'s are i.i.d with bounded variance. Using Kolmogorov's strong law of large numbers \cite{sen1994large},
    \[\hat{F}^\pi_n(\nu) = \frac{1}{n} \sum_{i=1}^n M_i \stackrel{\text{a.s.}}{\longrightarrow} \E[\frac{1}{n} \sum_{i=1}^n M_i] = F^\pi(\nu) \]
\end{proof}

\mwaygeneralization*
\begin{proof}
    The proof follows the same technique from the proof of Theorem \ref{thm:main_result}, where the appropriate slot density ratios must be matching with the corresponding terms of the reward CDF decomposition. We have,
    \[F_R(\nu) := \sum_{1 \leq k_1 < k_2, \dots < k_m \leq K} \psi_{k_{1:m}}(A^{k_{1:m}}, X, \nu)\]
    For ease of notation, let $Y_{k_i} = \frac{\pi(A^{k_i}|X)}{\mu_{k_i}(A^{k_i}|X)}$. The importance weight $G_m$ is then,
    \[G_m = \sum_{1 \leq k_1 < k_2, \dots < k_m \leq K} \left(\prod_{i=1}^{m}\frac{\pi(A^{k_i}|X)}{\mu_{k_i}(A^{k_i}|X)} - 1\right) + 1 = \sum_{1 \leq k_1 < k_2, \dots < k_m \leq K} \left(\prod_{i=1}^{m}Y_{k_i} - 1\right) + 1 \]
    which has ${K \choose m} + 1$ terms. For ease of notation denote the summation indices $1 \leq k_1 < k_2, \dots < k_m \leq K$ by $k_1:k_m$.
    \[
        \begin{split}
            &E[G_m \cdot \I\{R \leq \nu\}|A,X] = \left( \sum_{k_1:k_m} \left(\prod_{i=1}^{m}Y_{k_i} - 1\right) + 1 \right) \left(\sum_{k_1:k_m} \psi_{k_{1:m}}(A^{k_{1:m}}, X, \nu)\right) \\
            &= \underbrace{\sum_{k_1:k_m} \left(\prod_{i=1}^{m}Y_{k_i}\right) \psi_{k_{1:m}}(A^{k_{1:m}}, X, \nu)}_{\text{matching indices}} + \underbrace{\sum_{k_1:k_m} \left( \sum_{l_1:l_m \neq k_1:k_m} \left(\prod_{i=1}^{m}Y_{k_i} - 1\right) - 1 + 1\right)  \psi_{k_{1:m}}(A^{k_{1:m}}, X, \nu)}_{\text{non-matching indices}}
        \end{split}
    \]
    The matching indices are grouped together in the first term. The second term contains the \textit{remaining} terms of the importance weight $G_m$, namely the indices $l_1:l_m$ that were not used in the first term along with the extra $-1$ from the first term. Now steps similar to the proof of Theorem \ref{thm:main_result} complete the proof. Take an expectation over $A \sim \mu(.|X)$.
    It can be seen that the second term equals 0 since $\E_\mu\left[\left(\prod_{i=1}^{m}Y_{k_i}\right)\right] = 1$. The first term equals $\E_\pi[\I\{R \leq \nu\} | X]$, due to importance sampling. Taking an expectation over $X \sim d_X(.)$, by the total law of expectation,
    \[\E_\mu[G_m \cdot \I\{R \leq \nu\}] =  \E_\pi[\I\{R \leq \nu\} ] = F^\pi(\nu)\]
\end{proof}

\end{document}